\documentclass[11pt,english,twoside]{article}
\usepackage{babel}
\usepackage{blindtext}
\usepackage{hyperref}
\usepackage{graphicx}
\usepackage{fancyhdr}
\usepackage{amsmath,amssymb,graphicx}
\usepackage{xcolor}
\usepackage{titlesec}
\usepackage[export]{adjustbox}
\usepackage{subcaption}
\usepackage[utf8]{inputenc}
\usepackage{float}
\usepackage{mathtools}
\usepackage{txfonts}
\usepackage{pxfonts}
\usepackage[font=scriptsize,labelfont=bf]{caption}

\usepackage{amsthm}
\usepackage{float}
\usepackage{comment}
\graphicspath{ {images/} }
\usepackage[font=scriptsize,labelfont=bf]{caption}
\titleformat{\section}{\normalsize\bfseries}{\thesection}{1em}{}
\titleformat{\subsection}{\normalsize\bfseries}{\thesubsection}{1em}{}

\newtheorem{theorem}{Theorem}
\newtheorem{lemma}{Lemma}
\newtheorem{definition}{Definition}

\newcommand{\ASI}{\mathrm{ASI}}
\newcommand{\SLP}{\mathrm{SLP}}

\usepackage{geometry}
\begin{document}

\thispagestyle{empty}
\setcounter{page}{9}

\pagestyle{fancy}
\fancyhf{}
\fancyhead{}
\fancyhead[RO,LE]{\vspace{15pt}\\Computational Complexity of Determining the Assembly Index} 
\fancyfoot{}
\fancyfoot[LE,RO]{\thepage}
\fancyfoot[RE,LO]{\url{https://ipipublishing.org/index.php/ipil/}}
\renewcommand{\headrulewidth}{0.4pt} 

\begin{minipage}{0.14\textwidth}
\includegraphics[width=0.9\textwidth]{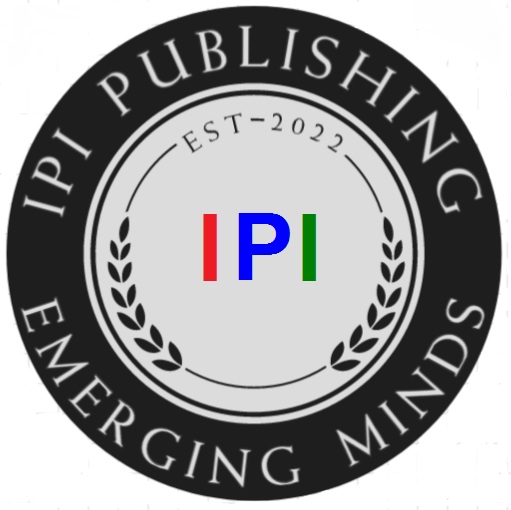} 
\end{minipage}
\hfill
\begin{minipage}{0.5\textwidth}
\includegraphics[width=1.05\textwidth]{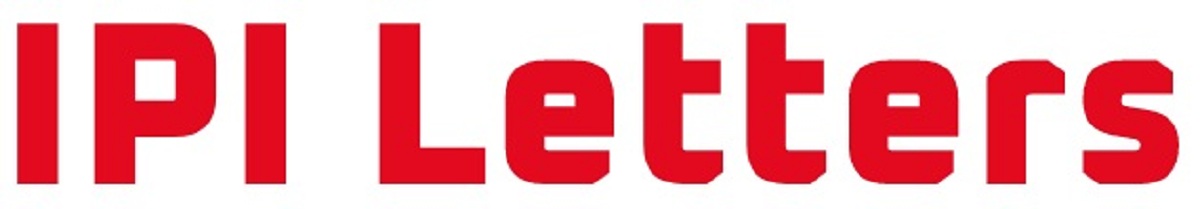}
\end{minipage}
\begin{minipage}{0.3\textwidth}
\begin{flushright}
{\scriptsize 
ISSN 2976 - 730X\\
IPI Letters 2026,Vol 4 (1):9-12\\
\href{https://doi.org/10.59973/ipil.xx}{\color{blue}{https://doi.org/10.59973/ipil.315}}\\
\medskip
Received: 2026-01-02\\
Accepted: 2026-01-15\\
Published: 2026-01-20\\
}\end{flushright}
\end{minipage}

\vspace{0.5cm}

\par\noindent\rule{\textwidth}{0.5pt}\\
{\color{red}\textbf{Article}} 

\begin{center}
\vspace{0.5cm}
  {\huge {\bf Computational Complexity of Determining the Assembly Index}}

  \vspace{0.5cm}
  
\end{center}

\noindent
{\large {\bf Piotr Masierak$^\bold{1,*}$ }

\vspace{0.1in}

\noindent
{\footnotesize $^1$Łukaszyk Patent Attorneys, ul. Głowackiego 8, 40-052 Katowice, Poland 

\vspace{0.1in}

\noindent
$^*$Corresponding author: \href{mailto:pmasierak@patent.pl} {\color{blue}{ pmasierak@patent.pl}}
}
\vspace{1cm}

\noindent
{\small {\bf Abstract} 
The assembly index of assembly theory quantifies the minimal number of composition steps required to construct an object from elementary components.
The study proves that the decision version of the assembly index problem is NP-complete, through an explicit correspondence between assembly plans and straight-line grammars.
This correspondence implies that the optimization version of the assembly index problem inherits NP- and APX-hardness from the classical smallest grammar problem.
The study provides complete, self-contained proofs for both decision and optimization variants of the assembly index problem.
These results establish that computing or approximating the assembly index is computationally intractable, placing it within the same complexity class as grammar-based compression.
}

\vspace{0.75cm}


\vspace{0.2cm}
\par\noindent\rule{\textwidth}{0.5pt}

\section{Introduction}

Assembly theory (AT) ~\cite{cronin2017, cronin2022, cronin2023} studies how complex structures emerge from simpler components via reuse of intermediates.
The assembly index (ASI) quantifies this generative complexity as the minimal number of joining steps required to construct a structure from elementary and intermediate components.
Intuitively, the ASI captures how deeply an object is embedded in a hierarchy of possible assembly pathways.
Although this concept has been widely discussed in theoretical and experimental contexts, its precise computational characterization has remained open.
It was conjectured \cite{lukaszyk_assembly_2024} that the problem of determining the assembly index is\footnote{Formally, a problem \textit{is not} a complexity class, but belongs to a complexity class.} NP-complete.
An attempt to prove this conjecture was offered in \cite{kempes_assembly_2025}. However, this proof was shown to violate the principles of AT by a predefined assembly space involving only a set of predefined assembly steps \cite{lukaszyk_assembly_2025}.
In this work, we consider the general, string-based formulation of the ASI
with binary concatenation and assembly spaces over strings with free terminals and unlimited reuse.

\section{Methods}

The smallest grammar problem (SGP) was shown to be NP-complete \cite{Charikar2005}, with additional results on APX-hardness and bounds on approximation ratios.
Furthermore, connections between SLPs and Lempel–Ziv factorization were analyzed \cite{Rytter2003} and the hardness results were extended to fixed alphabets and Chomsky normal form grammars \cite{casel_complexity_2021}.

We shall show that determining the ASI is computationally equivalent to finding the shortest straight-line program (SLP) generating $w$, a concept well known in grammar-based compression.
This equivalence allows us to transfer established hardness results from the smallest grammar problem (SGP) to the ASI, connecting these previously separate research areas.

We shall prove that the decision version of the ASI problem (\textsc{ASI-DEC}) is NP-complete and that its optimization version (\textsc{ASI-OPT}) is NP- and APX-hard.
These findings place the ASI within the same theoretical framework as grammar-based compression, showing that computing or approximating $\ASI(w)$ is computationally intractable unless $P=NP$.

\section{Results}

Let $\Sigma$ be a finite alphabet and $w \in \Sigma^+$ a word.
The initial pool contains all letters of $\Sigma$, the assembly step involves the concatenation ($\circ$) of exactly two existing words (letters or words previously formed), and each previously formed word may be reused arbitrarily many times.
The goal is to construct the word $w$, and the cost of achieving the goal is the number of assembly steps.
The $\ASI(w)$ is the minimal number of assembly steps required to construct the word $w$.

\begin{definition}[\textsc{ASI-DEC}]\label{def:ASI-DEC}
Given a pair $(w,k)$ with $w\in\Sigma^+$ and $k\in\mathbb{N}$, decide whether $\ASI(w)\le k$.
\end{definition}

The problem is existential: the input does not include the assembly plan but only asks whether such a plan exists. 
The NP witness corresponds to such a plan, and the verifier checks it in polynomial time by reconstructing the concatenation sequence. 
The distinction between the assembly steps problem and the assembly index problem is that the former concerns verifying a specific assembly sequence within a bounded number of steps, while the latter, formalized in Definition~\ref{def:ASI-DEC}, is existential, asking only whether such a sequence exists. Accordingly, its NP witness is the assembly plan itself.

\bigskip
A straight-line program (SLP) is a context-free grammar generating exactly one word $w$, with rules of the form $X \to YZ$ or $X \to a$. Its size is the number of concatenation rules $X \to YZ$.

\begin{definition}[\textsc{SLP-DEC}]\label{def:SLP-DEC}
Given a pair $(w,K)$, decide whether there exists an SLP of size less than or equal to $K$ generating $w$.  
This problem is NP-complete~\cite{Charikar2005, Rytter2003, casel_complexity_2021}.
\end{definition}

\begin{lemma}\label{lem:plan-slp}
For any word $w$, the minimal number of assembly steps equals the minimal number of concatenation rules in an SLP generating $w$, that is $\ASI(w) = \SLP(w)$.
\end{lemma}
\begin{proof}
Each assembly plan of $k$ concatenations corresponds to an SLP of size $k$, where every step $U \circ V$ becomes a production $X \to U V$.  
Conversely, every SLP of size $K$ yields a $K$-step assembly plan by expanding its rules in topological order.
Both count the same number of concatenations.
\end{proof}

For example, let $w = 01010$ over binary alphabet $\Sigma = \{0, 1\}$ (the initial pool).
An optimal assembly plan containing three steps $S_i$ may have the form
\begin{align*}
S_1 &\coloneqq 0 \circ 1 = 01;\\    
S_2 &\coloneqq S_1 \circ 0 = 01 \circ 0 = 010;\\
S_3 &\coloneqq S_2 \circ S_1 = 010 \circ 01 = 01010 = w.
\end{align*}
corresponding to the SLP of size $K=3$
\begin{align*}
R_1 &\to 0 \circ 1 \quad \text{(Rule 1 corresponds to step $S_1$)}; \\
R_2 &\to R_1 0   \quad \text{(Rule 2 corresponds to step $S_2$)}; \\
R_3 &\to R_2 R_1 \quad \text{(Rule 3 corresponds to step $S_3$)}. \\
\end{align*}

\begin{lemma}\label{lem:in-np}
\textsc{ASI-DEC} $\in$ NP.
\end{lemma}
\begin{proof}
Let the input be $(w,k)$, where $w \in \Sigma^{+}$ has length $n \coloneqq |w|$ and $k\in \mathbb{N}$ is encoded in binary.
We shall show that there exists a polynomial-time verifier whose running time is polynomial in the input size $n+\log_2 k$.

\medskip
\noindent
Any construction of a word of length $n$ from single letters using binary concatenations needs at most $n-1$ concatenation steps: starting from pieces of total length $1$, each concatenation increases the length of the newly created word by at least $1$, and the final length is $n$.
Hence, if $k \ge n-1$ then the instance is trivially YES. Otherwise, $k \le n-2$, so any feasible witness plan uses at most $t \le k=O(n)$ steps.

\medskip
\noindent
A witness is an assembly plan given as a sequence of $t$ steps $X_i \coloneqq U_i \circ V_i$ for $i=1,2,\dots,t$, where $t\le k$.
Each $U_i,V_i$ is specified \emph{by pointers/indices} either to:
(i) a terminal symbol $a \in \Sigma$ (constant-size encoding), or
(ii) one of the previously constructed words $X_j$ with $j<i$.
Thus, the witness size is $O(t\log_2 t)=O(n\log_2 n)$ bits.

\medskip
\noindent
The verifier reconstructs all intermediate strings $X_i$ in increasing order of $i$ and checks that the last constructed string equals $w$.
Each intermediate string has length at most $n$, hence the total verification time is $O(t\cdot n)\subseteq O(n^2)$ because $t\le k\le n-1$ in the nontrivial cases, which is polynomial in $n$ and hence also polynomial in the input size $n+\log_2 k$.
\end{proof}

\begin{theorem}\label{thm:np-complete-general}
$\textsc{ASI-DEC} \in \textsc{NP-complete}$ for any finite alphabet $\Sigma$.
\end{theorem}
\begin{proof}
From Lemma~\ref{lem:in-np} we know that \textsc{ASI-DEC} $\in$ NP.  
For NP-hardness, we reduce \textsc{SLP-DEC} to \textsc{ASI-DEC}: by Lemma~\ref{lem:plan-slp}, for every instance $(w, K)$, we define $(w, k \coloneqq K)$, which leads to
\begin{align*}
(w,K) &\in \text{\textsc{SLP-DEC}} 
\iff \SLP(w)\le K 
\iff \ASI(w)\le K 
\iff (w,k)\in \text{\textsc{ASI-DEC}} .
\end{align*}
The reduction is identity-based and computable in $O(|w|)$ time, proving the NP-hardness.
\end{proof}

The reduction is purely existential and does not require an explicit plan or grammar.
It shows that the existence of a grammar of size lower than or equal $K$ is equivalent to the existence of an assembly plan of cost lower than or equal $K$.
Therefore, determining $\ASI(w)$ lies in the same computational class as finding the smallest SLP.

Theorem~\ref{thm:np-complete-general} shows that determining whether a word can be assembled within $k$ concatenations is, therefore, in general, computationally intractable unless $P = NP$.

\bigskip
For a given word $w$, we define the ASI optimization version:
\begin{definition}[\textsc{ASI-OPT}]\label{def:ASI-OPT}
Given $w\in\Sigma^+$, minimize the number of its assembly steps to reach $\ASI(w)$.
\end{definition}

Let $\SLP(w)$ denote the minimal number of concatenation rules in an SLP generating $w$.

\begin{theorem}\label{thm:aiopt-hard}
The optimization version of the ASI problem given by Definition~\ref{def:ASI-OPT} is NP-hard and APX-hard.
Computing $\ASI(w)$ cannot be done in polynomial time, and no polynomial-time approximation scheme exists unless $P=NP$.
\end{theorem}
\begin{proof}
Let \textsc{SLP-OPT} denote the optimization problem of finding the smallest SLP generating $w$.  
\textsc{SLP-OPT} is NP-hard and APX-hard~\cite{Charikar2005, Rytter2003, casel_complexity_2021}.  
By Lemma~\ref{lem:plan-slp}, $\ASI(w)=\SLP(w)$ for all $w$, so the identity mapping $f(w)=w$ defines a cost-preserving polynomial reduction.  
Therefore, both NP- and APX-hardness results transfer directly to \textsc{ASI-OPT}.
\end{proof}

\section{Conclusions}

We have shown that determining the ASI is NP-complete (decision version) and NP-/APX-hard (optimization version) for an arbitrary finite alphabet.
The equality $\ASI(w)=\SLP(w)$ places ASI within the same complexity framework as the SGP.  
Consequently, computing or approximating $\ASI(w)$ is as difficult as finding the smallest grammar generating $w$.

\section*{Acknowledgements}

I thank my partners Wawrzyniec Bieniawski and Szymon Łukaszyk for their clarifications, formal corrections, and improvements.

\bibliographystyle{vancouver}
\bibliography{main}

\end{document}